\documentclass[runningheads]{llncs}
\usepackage{graphicx}
\usepackage{amsmath}

\begin{document}

\title{FPTAS for barrier covering problem with equal circles in 2D}
\titlerunning{FPTAS for barrier covering}
\author{Adil Erzin\inst{1,2} \and Natalya Lagutkina \inst{2}}
\authorrunning{Adil Erzin et al.}
\institute{Sobolev Institute of Mathematics, Novosibirsk, Russia,\\
\email{adilerzin@math.nsc.ru}, \and Novosibirsk
State University, Novosibirsk, Russia,\\
\email{lagutnat@yandex.ru}}

\maketitle

\begin{abstract}
In this paper, we consider a problem of covering a straight line segment by equal circles that are initially arbitrarily placed on a plane by moving their centers on a segment or on a straight line containing a segment so that the segment is completely covered, the neighboring circles in the cover are touching each other and the total length of the paths traveled by circles is minimal. The complexity status of the problem is not known.  We propose a $O(n^{2+c}/\varepsilon^2)$--time FPTAS for this problem, where $n$ is the number of circles and $c>0$ is arbitrarily small real.

\keywords{barrier coverage, mobile sensors, FPTAS}
\end{abstract}

\section{Introduction}
Wireless sensor networks (WSNs) consist of autonomous sensors, the operation time of which depends on the capacity of the batteries. The task of the WSN is to collect and transfer data. The time during which the network performs these functions is called a \emph{WSN's lifetime}. Energy efficient functioning of the WSN allows extending the lifetime. One of the important applications of the WSN is the monitoring of extended objects (borders, roads, pipelines, perimeters of buildings, etc.), which are often called the \textit{barriers}.

Monitoring can be carried out using both stationary \cite{Chen:07,Chen:08,Fan:14,Wu:17,Zhang:16,Zhao:18} and mobile sensors \cite{Andrews:17,Bar:15,Benk:15,Bhattacharya:08,Carmi:17,Cherry:17,Czyzowicz:10,Er:18,Fan:14,Saipulla:10,Thom:17}. The problem of energy efficient covering a segment with stationary sensors with adjustable circular sensing areas initially located on a segment is considered in \cite{Fan:14,Zhang:16}. In \cite{Fan:14} it is proved that such a problem is NP-hard.

The authors of \cite{Zhang:16} consider the formulation when the sensors are located at a short distance from the segment. It is proved that the problem is also NP-hard, and FPTAS with time complexity equals $O(n^3/\varepsilon^2)$ is proposed for it, where $n$ is the number of sensors and $0<\varepsilon<1$ is the accuracy.

In general, the initial arrangement of the sensors can be arbitrary. If the sensor is mobile, then the moving energy consumption is directly proportional to the length of the path traveled. In this case, the problem of barrier monitoring by mobile sensors with circular coverage areas is to move the sensors so that each barrier point is in the coverage area of at least one sensor, and the total length of the movement paths of the sensors involved in the covering is minimal \cite{Er:18}. This problem is usually denoted as MinSum \cite{Cherry:17}.

In the literature, two main classes of barrier monitoring problems with mobile devices are considered. The first class is a one-dimensional (1D) problem, when the sensors are initially on the line containing the segment to be covered \cite{Andrews:17,Bar:15,Benk:15,Carmi:17,Czyzowicz:10,Fan:14,Thom:17}. In the two-dimensional (2D) case, the problem is considered when the sensors are arbitrarily located on the plane \cite{Bhattacharya:08,Cherry:17,Er:18}. In \cite{Czyzowicz:10} it is proved that even for the one-dimensional case when the radii of the circles are different, the MinSum problem is NP-hard. In the case when the radii of the circles coincide, a $O(n^2)$--time algorithm is proposed.

The complexity status of the MinSum problem of covering the barrier with equal circles in 2D is open. At the same time, a $O(n^4)$--time algorithm for constructing an optimal solution in the $L_1$ metric, which is a $\sqrt{2}$-approximate solution in the Euclidean metric, is proposed in \cite{Cherry:17}. This solution has the property of preserving the order when the initial order of the circles involved in the covering after moving to the barrier does not change. This result was improved in \cite{Er:18}, where an algorithm for constructing an order preserving cover (OPC) was proposed, having the complexity $O(n^2)$.

The \cite{Benk:15} presents the results for 1D problems when the circles are different, and the centers of circles are initially outside the barrier. In the case when the sensors lie on one side of the segment, a $O(n^5/\varepsilon^2)$--time FPTAS is proposed, and for the case when the sensors are located on both sides of the segment, a $O(n^7/\varepsilon^3)$--time FPTAS is proposed. If the radii of the circles coincide, the problem is solved with  $O(n\log_2 n)$--time complexity \cite{Andrews:17}, which is an improvement of the result in \cite{Czyzowicz:10}.

The authors of the work \cite{Bhattacharya:08} consider the problem of covering a circle with mobile sensors with equal circular coverage areas. A PTAS with $O(n^4/\varepsilon)$--time complexity is proposed for the case when the circle is covered with equal evenly spaced disks. This result is improved in \cite{Carmi:17}, where a $O(n^{2+\varepsilon^{\prime}}/\varepsilon^{O(1)})$--time is proposed ($\varepsilon^{\prime}>0$).

In \cite{Bar:15}, the barrier is covered with circles, initially located on the line containing the barrier. This takes into account that the sensors consume energy on both movement and monitoring, and two problems are considered: (a) when the radii of the sensors are fixed, (b) when the radii of the sensors can be adjusted. In case (b), it is required to find the optimal radius for each circle, which is the sensor's coverage area. For case (b), a $O(n^5/\varepsilon^{4(1+\frac{1}{\alpha})})$--time FPTAS is proposed, where $\alpha\geq 1$. For the case of $\alpha=1$ in \cite{Fan:14}, a $O(n^2/\varepsilon)$--time FPTAS is proposed. These results are improved in \cite{Thom:17}, which considers the 1D problem of covering a segment with different circles, each of which has its own weight. A weighting factor was added in order to get closer to reality, when different sensors may need different amounts of energy for their work (for example, due to the wear and obsolescence of some of them). A $O(n^3/\varepsilon^2)$--time FPTAS is proposed for the case when the sensors are initially located on one side of the segment to be covered. For the general case, when the sensors are initially located on both sides of the segment, a $O(n^5/\varepsilon^3)$--time FPTAS is proposed. In this paper it was proved that the weighted 1D problem of covering a segment with different circles is NP-hard even when the sensors are initially located on one side of the barrier.

In this paper, the problem of covering the barrier, which is represented by a straight line segment, and mobile sensors -- by the centers of circles on a plane, is considered. We assume that the initial location of the sensors is known and they all have the same circular coverage areas. It is required to move the centers of the circles to a segment or a line containing a segment so as to cover the segment with \emph{touching} circles and the total length of the paths of movement of the sensors would be minimal. For this problem, we propose a $O(n^{2+c}/\varepsilon^2)$--time FPTAS, where $c>0$ is arbitrary small real. As far as we know, this is the first FPTAS for the problem under consideration.

The paper is organized as follows. Section 2 presents the mathematical formulation of the problem under consideration. Section 3 is devoted to the description of the FPTAS scheme and the proof of the complexity of the scheme. Section 4 summarizes this work.

 \section{Problem Formulation}
Let a \emph{barrier} in the form of a straight line segment of length $L>0$ and a set $S$ ($|S|=n$) of arbitrarily arranged equal circles (disks) be given on the plane. It is required to cover the barrier by moving the centers of the circles on the barrier or on the line containing the barrier. We introduce a coordinate system so that the segment is located on the x-axis between the points $(0,0)$ and $(L,0)$. If a point belongs to a barrier, then for simplicity of notation we will omit its ordinate, if it doesn't lead to confusion. Let the points on the plane $p_i=(x_i,y_i)$ be the initial coordinates of the centers of the circles, and $r_i= 1$ be the radius of the circle $i\in S$. Enumerate the circles from left to right according to the values $x_i$, $i=1,\ldots,n$.

\begin{definition}
The function $\hat{p}: S\to R^2$ is called a \emph{cover} if each point of the segment is within the coverage area of at least one sensor when the final position of the sensors are $\hat{p}_i=(\hat{x}_i,\hat{ y}_i)$, $i\in S$.
\end{definition}

Not all sensors need to be involved in the coverage. Let $C\subseteq S$ be a subset of sensors that participate in covering the barrier.

\begin{definition}
The function $\hat{p}$ is called an \emph{order preserving cover} (OPC) if $\hat{x}_i <\hat{x}_j$ if and only if $i<j$ for all $i,j\in C$.
\end{definition}

Let $d(p_i,\hat{p}_i)$ is the distance between the points $p_i$ and $\hat{p}_i$. Then the MinSum problem is to search for the function $\hat{p}^*$, which is the solution to the problem

\begin{equation}\label{eq:1}
cost(\hat{p}^*) =\min\limits_{\hat{p}} cost(\hat{p})= \min\limits_{\hat{p}}\sum\limits_{i=1}^n d(p_{i},\hat{p}_{i}).
\end{equation}

In the general case, when the radii of the circles are different, the problem (\ref{eq:1}) is NP-hard \cite{Czyzowicz:10}. In the case of identical circles, the complexity status of the problem is unknown \cite{Cherry:17}.

In this paper, we consider a MinSum problem with identical circles, whose centers, except, possibly, the centers of the left and right circles, need to be moved to a barrier in such a way that the neighbouring circles touch each other in the cover. The centers of the left and right circles in the covering can be moved on the extension of the segment, so that they touch the adjacent circles.

\begin{remark}
The requirement that the circles touch each other is a special case of a uniform placement of sensors. If we require that the distance between adjacent sensors be the same, but less than the diameter of the circle, then we obtain a covering not of the line, but of the strip.
\end{remark}

In the next section, we propose FPTAS for the MinSum problem.

 \section{FPTAS}
So, the radius of all circles is 1. Then in the required coverage, all circles, except perhaps the leftmost and rightmost, cover segments of length 2.

\begin{theorem}
For the problem (\ref{eq:1}) there is a $O(n^{2+c}/\varepsilon^2)$--time FPTAS, where $c>0$ is arbitrary small real.
\end{theorem}

\begin{proof}
Since the adjacent circles in the cover touch each other, then the cover depends on the length $\Delta\in(0,2]$ of the first (leftmost) segment, which is covered by one circle. Knowing $\Delta$, it is possible to unambiguously divide the barrier into subsegments (let's call them the \emph{cells}) of length 2 (perhaps, except for the first and last cells). As an auxiliary problem, we can solve a generalized assignment problem (\ref{eq2}). To formulate this problem, we number the cells from left to right $j=1,\ldots,m$ and denote by $c_{ij}=d(p_i,\hat{p}_i)$ the distance from the initial position of the center of the circle $i$ to its final position in order to cover cell $j$. For each cell except possibly the first and last, the final positions of the centers of the circle will coincide with the midpoints of the cell. For the first and last cells, the points $\hat{p}_i$ are closest to $p_i$ lying on the segment or on the line containing the barrier and the corresponding cell is covered. Then the generalized assignment problem has the following form:

\begin{equation}\label{eq2}
\left\{
\begin{array}{ccc}
W(\Delta)=\min\limits_{x_{ij}\in\{0,1\}}\sum\limits_{i=1}^n\sum\limits_{j=1}^m c_{ij}x_{ij} \\
\sum\limits_{i=1}^n x_{ij}=1, j=1,\dots,m \\
\sum\limits_{j=1}^m x_{ij}\leq 1, i=1,\dots,n \\
 \end{array}
\right.
 \end{equation}

However, the value of $\Delta$ (the length of the first cell) can take any value from the interval $(0,2]$. We introduce a grid on the interval $[0,2]$ with a step $\varepsilon^2$ and for each $\Delta=\varepsilon^2 k$, $k=1,\ldots,2/\varepsilon^2$ solve the generalized assignment problem (\ref{eq2}).

Let $W(\Delta^*)=cost(\hat{p}^*)$ be the minimal value of the problem functional. If $\Delta^*$ is a multiple of $\varepsilon^2$, then the optimal solution to (\ref{eq2}) is the optimal solution to (\ref{eq:1}). Suppose that $\Delta^*$ is not a multiple of $\varepsilon^2$, that is, does not coincide with $\varepsilon^2 k$ for any $k\in\mathcal{N}$. Denote by $\overline{\Delta}=\varepsilon^2 k^*$ the closest value to $\Delta^*$ that is at least $\Delta^*$. Then

$$|\varepsilon^2 k^*-\Delta^*|=\min\limits_k|\varepsilon^2 k-\Delta^*|\leq\varepsilon^2.$$

Shift all the circles of the optimal coverage by an equal distance right so that the first circle covers the segment $[0,\overline{\Delta}]$ and touches the second circle at the point $\overline{\Delta}$ (see Fig. 1). Note that in the approximate solution constructed as a result of solving the problem (\ref{eq2}) with $\Delta=\overline{\Delta}$, the segment $[\overline{\Delta},L]$ is optimally covered. Then

$$W(\overline{\Delta}) - W(\Delta^*) \leq\varepsilon^2$$
and
$$\frac{W(\overline{\Delta}) - W(\Delta^*)}{W(\Delta^*)} \leq \frac{\varepsilon^2}{W(\Delta^*)}.$$

\begin{figure}
\centering
\includegraphics[bb= 0 3 400 200, clip, scale=0.45]{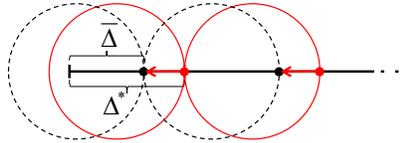}
\caption{Dotted black circles indicate an approximate solution, and red solid circles represent the optimal coverage.} \label{fig1}
\end{figure}

Two cases need to be considered:\\
(A) In the optimal coverage, at least one circle moves by at least $\varepsilon$ distance.\\
(B) In the optimal coverage, none of the circles moves a distance of $\varepsilon$ or more.

In the case (A), we have $W(\Delta^*)\geq\varepsilon$ and, therefore

$$\frac{W(\overline{\Delta}) - W(\Delta^*)}{W(\Delta^*)} \leq \frac{\varepsilon^2}{W(\Delta^*)}\leq\varepsilon.$$

In the case (B), circles horizontally move no more than a distance $\varepsilon$. This means that the centers of neighboring circles participating in the optimal coverage are initially located at a distance of at least $2-\varepsilon$ and at most $2+\varepsilon$. Thus, only the order preserving covering (OPC) can be optimal, and OPC can be constructed by the following dynamic programming algorithm $A_\varepsilon$.

\subsection{Algorithm $A_\varepsilon$}
Algorithm $A_\varepsilon$ should build a solution to problem (\ref{eq:1}), in which each circle moves a distance not exceeding the $\varepsilon$. Let $S_k(l)$ be the minimum total distance of displacements of the first $k$ circles participating in the covering of the segment $[0,l]$ containing of $k$ cells.

\subsubsection{Forward Recursion}
The initial position of the center of the first circle involved in the covering of the first cell of length $l\leq 2$ can only be inside the rectangle $R_1(l)$ with nodes in the plane points $(l-1-\varepsilon,0)$, $(l-1-\varepsilon,\varepsilon)$, $(l-1+\varepsilon,\varepsilon)$, $(l-1+\varepsilon,0)$. To cover the first segment $[0,l]$, one of the circles, whose center is in the rectangle $R_1(l)$, can move to the point $l-1$ (Fig. 2\emph{a}). As $l$ increases, rectangle $R_1(l)$ moves to the right, and centers of other circles will fall into it. These circles cannot be used to the right to cover the cell $k=2,3,\ldots$ (otherwise, there will be a movement greater than $\varepsilon$). Therefore, for each $l\leq 2$ the circle closest to point $l-1$, let it be the center of disk $N_1(l)$, is chosen, if it exists, and its movement determines $S_1(l)$. If the rectangle $R_1(l)$ contains no center of the circle, then $S_1(l)=+\infty$.

Consider the coverage of the segment $[0,l]$ which is split into $k\geq 2$ cells each of length 2 except maybe the first one. The centers of the circles that can cover the cell $k$ are inside the rectangle $R_k(l)$ with nodes in the points on the plane $(l-1-\varepsilon,0)$, $(l-1-\varepsilon,\varepsilon)$, $(l-1+\varepsilon,\varepsilon)$, $(l-1+\varepsilon,0)$. Let $Q_k(l)$ is the set of circle centers in the rectangle $R_k(l)$, and $d_i(l)$ is the distance from the vertex $i\in Q_k(l)$ to the point $l-1$ on the barrier. Then

$$
S_k(l)=\min\limits_{i\in Q_k(l)}d_i(l)+S_{k-1}(l-2)=d_{N_k(l)}(l)+S_{k-1}(l-2),
$$
where $l\in (0,L+2)$, $k=1,\ldots,K$ and the maximal number of cells is the integer $K\in [L/2,L/2+2)$. If $Q_k(l)$ is empty, then $S_k(l)=+\infty$.

The center of the last circle in the cover may be at the point of the straight line containing barrier belonging to the interval  $(L-1,L+1)$. This circle will cover the cell $[L-\delta,L]\subseteq[L-\delta,L-\delta +2]$, where $\delta\in [0,2)$. Since  $K\in [L/2,L/2+2)$ circles can participate in the coverage, to determine $cost(\hat{p}^*)$ it suffices to find

$$
S_k(l)==\min\limits_{K\in [L/2,L/2+2);\ l\in [L,L+2)}S_K(l)=S_{K^*}(l^*).
$$

\subsubsection{Backward Recursion}
During the backward recursion, the optimal OPC can be found using $N_k(l)$. After the forward recursion terminate, we know the optimal value of the objective function $cost(\hat{p}^*)$, the number of cells $K^*$, as well as the last cell $[L-\delta,L]$ and the the circle $N_{K^*}(l^*)$ that covers it. This means that the right border of the segment, namely $l^*-2$, which is covered by the first $K^*-1$ circles, is known. During the forward recursion, the best circle $N_{K^*-1}(l^*-2)$ was found to cover the cell $K^*-1$. It will cover a segment of length 2 and determine the right boundary of the barrier, which is covered with first $K^*-2$ circles. Continuing this process, we will build the desired cover.

\begin{definition}
If for any small $\delta>0$ we have that $S_k(l-\delta)=f_1(l-\delta)$, $S_k(l+\delta)=f_2(l+\delta)$ and $f_1(l)=f_2(l)$ or $S_k(l)=+\infty$, then $l$ is a \emph{switching point}.
\end{definition}

\begin{figure}
\centering
\includegraphics[bb= 0 0 700 300, clip, scale=0.5]{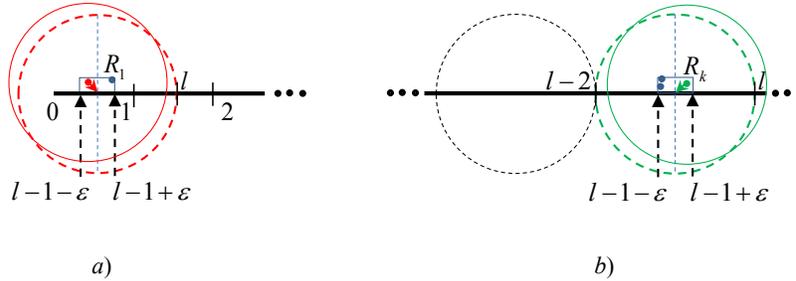}
\caption{$a$) Moving the first (solid red) circle participating in the covering to a new location (dotted circle); $b$) Moving an arbitrary circle (solid green) participating in the covering to a new location (dotted circle).} \label{fig1}
\end{figure}

Rectangle $R_k(l)$ may contain a finite set of circle centers. In any case, there are no more than $n$ of them there. With increasing $l$, the function determining the value of $S_k(l)$ changes a finite number of times in the switching points. The number of cells in the case of the existence of a solution is not more than $n$. Backward recursion is less time consuming. Therefore, the time complexity of building an optimal OPC by algorithm $A_\varepsilon$ is $O(n^2)$. The metric generalized assignment problem for each $\Delta$ is solved with time complexity $O(n^{2+c})$, where $c$ is arbitrarily small positive number \cite{Agarwal:06}. The number of the generalized assignment problems is limited to $2/\varepsilon^2$. As a result, the best of the solutions built by two algorithms (dynamic programming builds an optimal OPC when the length of movement of each circle does not exceed $\varepsilon$, and by solving assignment problems (\ref{eq2}) for different $\Delta$ another solution is found) has a relative error of no more than $\varepsilon$. The complexity of constructing the OPC is $O(n^2)$, and the complexity of solving the problem (\ref{eq2}) is $O(n^{2+c})$ \cite{Agarwal:06}. Then the total time complexity of FPTAS is $O(n^{n+c}/\varepsilon^2)$, QED.
\end{proof}

\section{Conclusion}

In this paper, we consider the problem of covering a straight line segment by arbitrary located equal circles on a plane by moving their centers to a segment or perhaps its continuation (for the leftmost and rightmost circles) in such a way that the neighboring circles involved in the covering touch each other and the total length of the circles movements would be minimal. The complexity status of this problem is not known. We propose a $O(n^{2+c}/\varepsilon^2)$--time FPTAS, where $n$ is the number of circles and $c>0$ is arbitrarily small real. As far as we know, this is the first FPTAS for the problem under consideration.

\section*{Acknowledgements}

The research is partly supported by the program of fundamental scientific researches of the SB RAS (project 0314--2019--0014) and by the Russian Foundation for Basic Research (Projects 17--51--45125).


\begin{thebibliography}{25}

\bibitem{Agarwal:06} Agarwal P.K., Efrat A., Sharir M.: Vertical decomposition of shallow levels in 3-dimensional arrangements and its applications. SIAM J. Comput., 29(3), pp. 912-–953 (2006)

\bibitem{Andrews:17} Andrews A.M., Wang H.: Minimizing the aggregate movements for interval coverage. Algorithmica, 78(1), pp. 47--85 (2017)

\bibitem{Bar:15} Bar-Noy A., Rawitz D. and Terlecky P.: ``Green'' barrier coverage with mobile sensors. In: Algorithms and Complexity: 9th International Conference, CIAC 2015, Paris, France. Proceedings, pp. 33--46. Springer International Publishing (2015)

\bibitem{Benk:15} Benkoczi R., Friggstad Z., Gaur D. and Mark Thom M.: Minimizing total sensor movement for barrier coverage by non-uniform sensors on a line.In Proc. of the 11th International Symposium on Algorithms and Experiments for Wireless Sensor Networks. LNCS vol 9536, pp.98--111 (2015)

\bibitem{Bhattacharya:08} Bhattacharya B.K., Burmester M., Hu Y., Kranakis E., Shi Q., and A. Wiese A.: Optimal movement of mobile sensors for barrier coverage of a planar region. Theor. Comput. Sci., 410(52), pp. 5515--5528 (2008)

\bibitem{Carmi:17} Carmi P., Katz M.J., Saban R., Stein Y.: Improved PTASs for convex barrier coverage. In: Solis-Oba R., Fleischer R. (eds) Approximation and Online Algorithms. WAOA 2017. LNCS 10787, pp. 26--40. Springer, Cham (2017)

\bibitem{Chen:07} Chen A., Kumar S., Lai T.H.: Designing localized algorithms for barrier coverage. In: Proceedings of the 13th annual ACM international conference on Mobile computing and networking, pp. 63--74 (2007)

\bibitem{Chen:08} Chen A., Lai T.H., Xuan D.: Measuring and guaranteeing quality of barrier-coverage in wireless sensor networks. In: Proceedings of the ACM International Symposium on Mobile Ad Hoc Networking and Computing (MobiHoc), pp. 421--430 (2008)

\bibitem{Cherry:17} Cherry A., Gudmundsson J.,  Mestre J.: Barrier coverage with uniform radii in 2D. In: Fernandez Anta A. et al. (eds.) ALGOSENSORS 2017, LNCS 10718, pp. 57--69 (2017)

\bibitem{Czyzowicz:10} Czyzowicz J. et al.: On Minimizing the sum of sensor movements for barrier coverage of a line segment. In: Nikolaidis I., Wu K. (eds) Ad-Hoc, Mobile and Wireless Networks. Ad Hoc, Mobile and Wireless Networks. LNSC 6288, pp. 29--42. Springer, Berlin, Heidelberg (2010)

\bibitem{Er:18} Erzin A., Lagutkina N.: Barrier coverage problem in 2D. In: S. Gilbert  et al. (Eds.): ALGOSENSORS 2018, LNCS 11410, pp. 118--130. Springer, Berlin, Heidelberg (2019)

\bibitem{ErLagIr:19} Erzin A., Lagutkina N., Ioramishvili N.: Barrier covering in 2D using mobile sensors with circular coverage areas. Arxiv (2019)

\bibitem{Fan:14} Fan H., Li M., Sun X., Wan P., Zhao Y.: Barrier coverage by sensors with adjustable ranges. ACM Transactions on Sensor Networks 11(1), Article No. 14. (2014)

\bibitem{Saipulla:10} Saipulla A., Westphal C., Liu B., Wang J.:  Barrier coverage with line-based deployed mobile sensors. Ad Hoc Networks, pp. 1381--1391 (2010)

\bibitem{Thom:17} Thom M.: Investigation on two classes of covering problems // PhD thesis, University of Lethbridge, Canada (2017)


\bibitem{Wu:17} Wu F., Gui Y., Wang Z., Gao X., Chen G.: A survey on barrier coverage with sensors. Front. Comput. Sci., 10(6), pp. 968--984 (2016)

\bibitem{Zhang:16} Zhang X.: Algorithms for barrier coverage with wireless sensors // PhD thesis, City University of Hong Kong, Hong Kong (2016)

\bibitem{Zhao:18} Zhao L., Bai G., Shen H., Tang Z.: Strong barrier coverage of directional sensor networks with mobile sensors. Int. J. of Distributed Sensor Networks, 14(2) (2018) DOI: 10.1177/1550147718761582.


\end{thebibliography}
\end{document}